\documentclass
{statsoc}
\usepackage{fancyhdr}
\usepackage{amsmath}
\usepackage{amssymb}
\usepackage{graphicx}
\usepackage{amsfonts}
\usepackage{float}
\usepackage{url,graphicx,tabularx,array,geometry,amsmath,appendix}

\newtheorem{theorem}{Theorem}
\newtheorem{lemma}[theorem]{Lemma}

\newtheorem{corollary}[theorem]{Corollary}

\title{Bayesian Conditional Tensor Factorizations for High-Dimensional Classification}
\author[Yun Yang {\it et al.}]{Yun Yang}
\address{Department of Statistical Science, Duke University,
         Durham,
         USA.}
\email{yy84@stat.duke.edu}
\author{David B. Dunson}
\address{Department of Statistical Science, Duke University,
         Durham,
         USA.}
\email{dunson@stat.duke.edu}

\begin{document}
\bibliographystyle{chicago}

\maketitle
\begin{abstract}
In many application areas, data are collected on a categorical response and high-dimensional
categorical predictors, with the goals being to build a parsimonious model for classification while doing inferences on the important predictors.
In settings such as genomics, there can be complex interactions among the predictors. By using a
carefully-structured Tucker factorization, we define a model that can characterize any
conditional probability, while facilitating variable selection and modeling of higher-order
interactions. Following a Bayesian approach, we propose a Markov chain Monte Carlo
algorithm for posterior computation accommodating uncertainty in the predictors to be
included. Under near sparsity assumptions, the posterior distribution for the conditional
probability is shown to achieve close to the parametric rate of contraction even in ultra
high-dimensional settings. The methods are illustrated using simulation examples
and biomedical applications.
\end{abstract}

\section{Introduction}

Classification problems involving high-dimensional categorical predictors have become common in a variety of application areas, with the goals being not only to build an accurate classifier but also to identify a sparse subset of important predictors.  For example, genetic epidemiology studies commonly focus on relating a categorical disease phenotype to single nucleotide polymorphisms  encoding whether an individual has 0, 1 or 2 copies of the minor allele at a large number of loci across the genome.  In such applications, it is expected that interactions play an important role, but there is a lack of statistical methods for identifying important predictors that may act through both main effects and interactions from a high-dimensional set of candidates.  Our goal is to develop nonparametric Bayesian methods for addressing this gap.

There is a rich literature on methods for prediction and variable selection from high or ultra high-dimensional predictors with a categorical response.  The most common strategy would rely on logistic regression with the linear predictor having the form $x_i'\beta$, with $x_i = (x_{i1},\ldots,x_{ip})'$ denoting the predictors and $\beta = (\beta_1,\ldots,\beta_p)'$ regression coefficients.  In high-dimensional cases in which $p$ is the same order of $n$ or even $p>n$, classical methods such as maximum likelihood break down but there is a rich variety of alternatives ranging from penalized regression to Bayesian variable selection.  Popular methods include $L_1$ penalization \citep{Tibshirani1996} and the elastic net  \citep{Zou2005}, which combines $L_1$ and $L_2$ penalties to accommodate $p\gg n$ cases and allow simultaneous selection of correlated sets of predictors.  For efficient $L_1$ regularization in generalized linear models including logistic regression, \cite{Park2007} proposed a solution path method.  \cite{Alexander2007} propose a related Bayesian approach for high-dimensional logistic regression under Laplace priors.  \cite{Tong2009} applied $L_1$ penalized logistic regression to genome wide association studies.  Potentially, related methods can be applied to identify main effects and epistatic interactions  \citep{Yang2010}, but direct inclusion of interactions within a logistic model creates a daunting dimensionality problem limiting attention to low-order interactions and modest numbers of predictors.

These limitations have motivated a rich variety of nonparametric classifiers, including classification and regression trees (CART) \citep{Breiman1984} and random forests (RFs)  \citep{Breiman2001}. CART partitions the predictor space so that samples within the same partition set have relatively homogeneous outcomes. CART can capture complex interactions and has easy interpretation, but tends  to be unstable computationally and lead to low classification accuracy. RFs extend CART by creating a classifier consisting of a collection of trees that are all used to vote for classification. RFs can substantially reduce variance compared to a single tree and result in high classification accuracy, but provides an
 uninterpretable black box that does not yield insight into the relationship between specific predictors and the outcome.  Moreover, through our simulation results in section 6, we found that random forests did not behave well in high dimensional low signal-to-noise cases.

Our focus is on developing a new framework for nonparametric Bayes classification through tensor factorizations of the conditional probability $P(Y=y\, |\, X_1=x_1,\ldots, X_p=x_p)$, with $Y \in \{1,\ldots,d_0\}$ a categorical response and $X=(X_1,\ldots,X_p)'$ a vector of $p$ categorical predictors.  The conditional probability can be expressed as a $d_1 \times \cdots \times d_p$ tensor for each class label $y$, with $d_j$ denoting the number of levels of the $j$th categorical predictor $X_j$.  If $p=2$ we could use a low rank matrix factorization of the conditional probability, while in the general $p$ case we could consider a low rank tensor factorization. Such factorizations must be non-negative and constrained so that the conditional probabilities add to one for each possible $X$, and are fully flexible in characterizing the classification function for sufficiently high rank.  \cite{Dunson2009} and \cite{Anirban2011} applied two different tensor decomposition methods to model the joint probability distribution for multivariate categorical data. Although an estimate of the joint pmf can be used to induce an estimate of the conditional probability, there are clear advantages to bypassing the need to estimate the high-dimensional nuisance parameter corresponding to the marginal distribution of $X$.

We address such issues using a Bayesian approach that places a prior over the parameters in the factorization, and provide strong theoretical support for the approach while developing a tractable algorithm for posterior computation.  Some advantages of our approach include (i) fully flexible modeling of the conditional probability allowing any possible interactions while favoring a parsimonious characterization; (ii) variable selection; (iii) a full probabilistic characterization of uncertainty providing measures of uncertainty in variable selection and predictions; and (iv) strong theoretical support in terms of rates at which the full posterior distribution for the conditional probability {\em contracts} around the truth.  Notably, we are able to obtain near a parametric rate even in ultra high-dimensional settings in which the number of candidate predictors increases exponentially with sample size.  Such a result differs from frequentist convergence rates in characterizing concentration of the entire posterior distribution instead of simply a point estimate.  Similar contraction rate results in $p$ diverging with $n$ settings are currently only available in simple parametric models, such as the normal means problem \citep{Castillo2012} and generalized linear models \citep{Jiang2006}.  Although our computational algorithms do not yet scale to massive dimensions, we can accommodate $1,000$s of predictors.

\section{Conditional Tensor Factorizations}

\subsection{Tensor factorization of the conditional probability}
Although there is a rich literature on tensor decompositions, little is in statistics.  The focus has been on two factorizations that generalize matrix singular value decomposition (SVD). The most popular is parallel factor analysis (PARAFAC) \citep{Harshman1970, Harshman1994, Zhang2001}, which expresses a tensor as a sum of $r$ rank one tensors, with the minimal possible $r$ defined as the rank. The second approach is Tucker decomposition or higher-order singular value decomposition (HOSVD), which was proposed by \cite{Tucker1966} for three-way data and extended to arbitrary orders by \cite{DeLathauwer2000}. HOSVD expresses a $d_1 \times \cdots \times d_p$ tensor $A = \{ a_{c_1\cdots c_p} \}$ as
\begin{eqnarray}
a_{c_1\cdots c_p} = \sum_{h_1=1}^{d_1} \cdots \sum_{h_p=1}^{d_j} g_{h_1 \cdots h_p}\prod_{j=1}^p u_{h_jc_j}^{(j)}, \nonumber
\end{eqnarray}
where $G = \{ g_{h_1 \cdots h_p} \}$ is a core tensor, with constraints on $G$ such as low rank and sparsity imposed to induce better data compression and fewer components compared to PARAFAC. For probability tensors, we need nonnegative versions of such decompositions and the concept of rank changes accordingly \citep{Cohen1993}.

The conditional probability $P(Y=y|X_1=x_1,\ldots,X_p=x_p)$ can be structured as a $d_0 \times d_1\times \cdots \times d_p$ dimensional tensor. We will call such tensors {\em conditional probability tensors}. Let $\mathcal{P}_{d_1,\ldots,d_p}(d_0)$ denote the set of all conditional probability tensors, so that $P\in\mathcal{P}_{d_1,\ldots,d_p}(d_0)$ implies
\[
P(y|x_1,\ldots,x_p)\geq0\ \forall y, x_1,\ldots,x_p,\quad \sum_{y=1}^{d_0}P(y|x_1,\ldots,x_p)=1\ \forall x_1,\ldots,x_p.
\]
To ensure that $P$ is a valid conditional probability, the elements of the tensor must be non-negative with constraints on the first dimension for $Y.$  A primary goal is accommodating
high-dimensional covariates, with the overwhelming majority of cells in the table corresponding to unique combinations of $Y$ and $X$ unoccupied.  In such settings, it is necessary to
encourage borrowing information across cells while favoring sparsity.

Our proposed model for the conditional probability has the form:
      \begin{eqnarray}
        P(y|x_1,\ldots,x_p) &=& \sum_{h_1=1}^{k_1}\cdots\sum_{h_p=1}^{k_p}\lambda_{h_1h_2\ldots h_p}(y)\prod_{j=1}^p\pi_{h_j}^{(j)}(x_j),\label{eq:1}
      \end{eqnarray}
      with the parameters subject to
      \begin{eqnarray}\label{eq:2}
        &&\nonumber\sum_{c=1}^{d_0}\lambda_{h_1h_2\ldots h_p}(c)=1,\text{ for any possible combination of }(h_1,h_2,\ldots,h_p),\\
        &&\sum_{h=1}^{k_j}\pi_{h}^{(j)}(x_j)=1,\text{ for any possible pair of }(j,x_j).
      \end{eqnarray}
 The $k_j$ value controls the number of parameters used to characterize the impact of the $j$th predictor.  In the special case in which $k_j=1$, the $j$th predictor is excluded from the model, so sparsity can be imposed by setting $k_j=1$ for most $j$'s.  The representation (\ref{eq:1}) is many-to-one and the different parameters in the factorization cannot be uniquely identified. This does not present a barrier to our Bayesian approach and indeed over-parameterized models often have computational advantages in leading to simplified posterior computation and reduced autocorrelation in Markov chain Monte Carlo (MCMC) samples of parameters of interest, such as the induced predictive distribution (\cite{Bhattacharya2011}, \cite{Ghosh2009}).

We format the conditional probability $P(y|x_1,\ldots,x_p)$ as a $d_1\times\cdots\times d_p$ vector
\begin{eqnarray*}
Vec\{P(y|-)\}=\big\{P(y|1,\ldots,1,1),P(y|1,\ldots,1,2),\ldots,P(y|1,\ldots,1,d_p),
\ldots,\\
P(y|1,\ldots,d_{p-1},d_p),\ldots,P(y|d_1,\ldots,d_{p-1},d_p)\big\}'
\end{eqnarray*}
and $\lambda_{h_1,\ldots,h_p}(y)$ as a $k_1\times\cdots\times k_p$ vector
\begin{eqnarray*}
\lefteqn{ Vec\{\Lambda(y)\}=\big\{ \lambda_{1,\ldots,1,1}(y),\lambda_{1,\ldots,1,2}(y),\ldots, } \\
& & \qquad \qquad \quad \lambda_{1,\ldots,1,k_p}(y),
\ldots,\lambda_{1,\ldots,k_{p-1},k_p}(y),\ldots,\lambda_{k_1,\ldots,k_p}(y)\big\}'.
\end{eqnarray*}
Let $\pi^{(j)}$ be a $d_j\times k_j$ matrix with $\pi^{(j)}_v(u)$ as the $(u,v)$th element. It is a stochastic matrix, so rows sum to one, by constraint (\ref{eq:2}). Then representation (\ref{eq:1}) can be written in vector form:
\begin{equation}\label{eq:11}
Vec\{P(y|-)\}=\big(\pi^{(1)}\otimes\pi^{(2)}\otimes\cdots\otimes\pi^{(p)}\big)Vec\{\Lambda(y)\}, \text{ for }y=1,\ldots,d_0,
\end{equation}
where $\otimes$ denotes the Kronecker product. Furthermore, if we let $Mat(P)$ and $Mat(\Lambda)$ be two stochastic matrices with the $y$th column $Vec\{P(y|-)\}$ and $Vec\{\Lambda(y)\}$ respectively for $y=1,\ldots,d_0$, then we can write the above $d_0$ identities together as:
\[
Mat(P)=\big(\pi^{(1)}\otimes\pi^{(2)}\otimes\cdots\otimes\pi^{(p)}\big) Mat(\Lambda).
\]
The following theorem provides basic support for factorization (\ref{eq:1})-(\ref{eq:2}) through showing that any conditional probability has this representation. The proof of this theorem, which can be found in the appendix, sheds some light on the meaning of $k_1,\ldots,k_p$ and how it is related to a sparse structure of the tensor.

\begin{theorem}\label{thm:1}
Every $d_0\times d_1\times d_2\times\cdots\times d_p$ conditional probability tensor $P\in\mathcal{P}_{d_1,\ldots,d_p}(d_0)$ can be decomposed as (\ref{eq:1}),
      with $1\leq k_j\leq d_j$ for $j=1,\ldots,p$. Furthermore, $\lambda_{h_1h_2\ldots h_p}(y)$ and $\pi_{h_j}^{(j)}(x_j)$ can be chosen to be nonnegative and satisfy the constraints (\ref{eq:2}).
\end{theorem}

We can simplify the representation through introducing $p$ latent class indicators $z_1,\ldots,z_p$ for $X_1,\ldots,X_p$, with $Y$ conditionally independent of $(X_1,\ldots,X_p)$ given
 $(z_1,\ldots,z_p)$. The model can be written as
\begin{eqnarray}
  Y_i|z_{i1},\ldots,z_{ip} & \sim & \text{Multinomial}\big(\{1,\ldots,d_0\},\lambda_{z_{i1},\ldots,z_{ip}}\big), \nonumber \\
  z_{ij}|X_j &\sim & \text{Multinomial}\big(\{1,\ldots,k_j\}, \pi_1^{(j)}(X_j),\ldots,\pi_{k_j}^{(j)}(X_j)\big), \label{eq:3}
\end{eqnarray}
where $\lambda_{z_{i1},\ldots,z_{ip}}=
\big\{ \lambda_{z_{i1},\ldots,z_{ip}}(1),\ldots,\lambda_{z_{i1},\ldots,z_{ip}}(d_0)\big\}$. Marginalizing out the latent class indicators, the conditional probability of $Y$ given $X_1,\ldots,X_p$ has the form in (\ref{eq:1}).

\subsection{Prior specification}
To complete a Bayesian specification of our model, we choose independent Dirichlet priors for the parameters $\Lambda=\{\lambda_{h_1,\ldots,h_p},h_j=1,\ldots,k_j,j=1,\ldots,p\}$ and $\pi=\{\pi_{h_j}^{(j)}(x_j), h_j=1,\ldots,k_j,x_j=1,\ldots,d_j,j=1,\ldots,p\}$,
\begin{eqnarray}
\big\{ \lambda_{h_1,\ldots,h_p}(1),\ldots,\lambda_{h_1,\ldots,h_p}(d_0)\big\} & \sim & \text{Diri}(1/d_0,\ldots,1/d_0), \nonumber \\
\big\{ \pi_{1}^{(j)}(x_j),\ldots,\pi_{k_j}^{(j)}(x_j)\big\} & \sim & \text{Diri}(1/k_j,\ldots,1/k_j), j=1,\ldots,p. \label{eq:prior}
\end{eqnarray}
These priors have the advantages of imposing non-negative and sum to one constraints, while leading to conditional conjugacy in posterior computation.  The hyperparameters in the Dirichlet priors are chosen to favor placing most of the probability on a few elements, inducing near sparsity in these vectors.

If $k_j=1$ in (\ref{eq:1}), by constraints (\ref{eq:2}) $\pi_{1}^{(j)}(x_j)= 1$, $P(y|x_1,\ldots,x_p)$ will not depend on $x_j$ and $Y \perp X_j | X_{j'}, j' \neq j.$  Hence, $I(k_j>1)$ are variable selection indicators.  In addition, $k_j$ can be interpreted as the number of latent classes for the $j$th covariate. Levels of $X_j$ are clustered according to their relations with the response variable in a soft probabilistic manner, with $k_1,\ldots,k_p$ controlling the complexity of the latent structure as well as sparsity.

To embody our prior belief that only a small number of $k_j$'s are greater than one, we let
\[
P(k_j=1)=1-\frac{r}{p},\ P(k_j=k)=\frac{r}{(d_j-1)p}, \text{ for }k=2,\ldots,d_j, j=1,\ldots,p,
\]
where $r$ is the expected number of predictors included. To further impose sparsity, we include a restriction that $\sharp\{j:k_j>1\}\leq \bar{r}$, where $\bar{r}$ is a prespecified maximum number of predictors.  We can choose the upper bound to correspond to twice the number of predictors we expect to be important, though in practice results tend to be robust to these hyperparameters unless $\bar{r}$ is chosen to be too small.  Default values of $r$ and $\bar{r}$ based on theoretical considerations are suggested in section 3.3.

The effective prior on the $k_j$'s is
\begin{eqnarray}
P(k_1=l_1,\ldots,k_p=l_p)=P(k_1=l_1)\cdots P(k_p=l_p)I_{\{\sharp\{j:l_j>1\}\leq \bar{r}\}}(l_1,\ldots,l_p), \label{eq:priork}
\end{eqnarray}
where $I_A(\cdot)$ is the indicator function for set $A$.  Let $\gamma = (\gamma_1,\ldots,\gamma_p)'$ be a vector having elements $\gamma_j = I(k_j>1)$ indicating inclusion of the $j$th predictor.  Under prior (\ref{eq:priork}) the induced prior for $\gamma$ is equivalent to the prior in \cite{Jiang2006}.
Potentially, we can put a more structured prior on the components in the conditional tensor factorization, including sparsity in $\Lambda$.  However, the theory shown in the next section provides strong support for prior (\ref{eq:prior})-(\ref{eq:priork}).

\section{Properties}

\subsection{Bias-variance trade off}

Because we are faced with extreme data sparsity in which the vast majority of combinations of $Y,X_1,\ldots,X_p$ are not observed, it is critical to impose sparsity assumptions.  Even if such assumptions do not hold, they have the effect of massively reducing the variance, making the problem tractable.  A sparse model that discards predictors having less impact and parameters having small values may still explain most of the variation in the data, resulting in a useful classifier that has good performance in terms of the bias-variance tradeoff even when sparsity assumptions are not satisfied.  We provide a simple illustrative example to demonstrate the tendency of our model to produce low MSE.

Suppose we have a binary response $Y$ and $p$ binary covariates $X_j \in \{-1,1\}$, $j=1,\ldots,p$.  The true model can be expressed in the form
\begin{equation}\label{eq:9}
P(Y=1|X_1=x_1,\ldots,X_p=x_p)=\frac{1}{2}+\frac{\beta}{2^2}x_1+\cdots+\frac{\beta}{2^{p+1}}x_p,\  \beta\in(0,1).
\end{equation}
The effect of $X_j$ on the response $Y$ decreases exponentially as $j$ increases from $1$ to $p$. A natural strategy is to estimate
$P(Y=1|X_1=x_1,\ldots,X_p=x_p)$ by the sample frequencies over the first $k$ covariates $\hat{P}(Y=1|X_1=x_1,\ldots,X_k=x_k)=\sharp\{i:y_i=1,x_{1i}=x_1,\ldots,x_{ki}=x_k\}/
\sharp\{i:x_{1i}=x_1,\ldots,x_{ki}=x_k\}$ and ignore the remaining $p-k$ covariates. Suppose we have $n=2^{l}$ ($k\leq l\ll p$) observations with one in each cell of combinations of $X_1,\ldots,X_l$.  Under (\ref{eq:9}), we can calculate the MSE of such an estimate as a function of $k$, and find a minimal MSE $k$ value.
\begin{eqnarray*}
\text{MSE}&=&\sum_{h_1,\ldots,h_p}E\big\{ P(Y=1|X_1=h_1,\ldots,X_p=h_p) - \\
& & \qquad \qquad \ \hat{P}(Y=1|X_1=h_1,\ldots,X_k=h_k)\big\}^2\\
&\triangleq&\text{Bias}^2+\text{Var}, \\
\text{Bias}^2&=&\sum_{h_1,\ldots,h_p}\big\{ P(Y=1|X_1=h_1,\ldots,X_p=h_p)- \\
& & \qquad \qquad  E\hat{P}(Y=1|X_1=h_1,\ldots,X_k=h_k)\big\}^2\\
&=&\beta^22^{k+1}\sum_{i=1}^{2^{p-k-1}}\bigg(\frac{2i-1}{2^{p+1}}\bigg)^2=\frac{\beta^2}{3}(2^{p-2k-2}-2^{-p-2}), \\
\text{Var}&=&\sum_{h_1,\ldots,h_p}\text{Var}\hat{P}(Y=1|X_1=h_1,\ldots,X_k=h_k)\\
&=&2^{p-k+1} \sum_{i=1}^{2^{k-1}}\frac{1}{2^l}\bigg(\frac{1}{2}+\frac{2i-1}{2^{k+1}}\beta\bigg)
\bigg(\frac{1}{2}-\frac{2i-1}{2^{k+1}}\beta\bigg)\\
&=&\frac{1}{3}\big\{ (3-\beta^2)2^{p+k-l-2}+\beta^22^{p-k-l-2}\big\}.
\end{eqnarray*}
Since there are $2^p$ cells, the average MSE for each cell equals
\[
\frac{1}{3}\big\{ (3-\beta^2)2^{k-l-2}+\beta^22^{-k-l-2}+\beta^22^{-2k-2}-\beta^22^{-2p-2}\big\}.
\]
From this we can see that $p$, the number of covariates, has little impact on the selection of $k$. Recall that $k\leq l$ and so the second term will be small comparing to the first and third terms. Hence, the average MSE obtains its minimum at $k\approx l/3=\log_2(n)/3$.  Even though the true model (\ref{eq:9}) is not sparse and all the predictors impact the conditional probability, the optimal number of predictors only depends on the log sample size, with the number of predictors playing almost no role.  This example also gives some intuition on the assumption of $r_n\sim\log(n)$ on the true model used in section 3.3.

\subsection{Borrowing of information}
A critical feature of our model is borrowing of information across cells corresponding to each combination of $X_1,\ldots,X_p$. Letting $w_{h_1,\ldots,h_p}(x_1,\ldots,x_p)= \prod_j \pi_{h_j}^{(j)}(x_j)$, model (\ref{eq:1}) is equivalent to
\begin{equation}
P(Y=y|X_1=x_1,\ldots,X_p=x_p)=\sum_{h_1,\ldots,h_p}w_{h_1,\ldots,h_p}(x_1,\ldots,x_p)\lambda_{h_1\dots h_p}(y), \nonumber
\end{equation}
and constraints (\ref{eq:2}) imply $\sum_{h_1,\ldots,h_p}w_{h_1,\ldots,h_p}(x_1,\ldots,x_p)=1$.
If $\lambda_{h_1\dots h_p}(y)$ is viewed as the frequency of $Y=y$ for the observations in cell $X_1=h_1,\ldots,X_p=h_p$, then our model essentially uses a kernel estimate that allows borrowing of information across cells via a weighted average of the cell frequencies.

To illustrate the strength of this, consider a simple example involving one covariate X with $m$ categories and a binary response. Let $P_j=P(Y=1|X=j)$ for $j=1,\ldots,m$. A naive estimate for $(P_1,\ldots,P_m)$ is sample frequencies $(k_1/n_1,\ldots,k_m/n_m)$, denoted by $(\hat{P}_1,\ldots, \hat{P}_m)$, where $k_j=\sharp\{i:y_i=1 \text{ and } x_i=j\}$ and $n_j=\sharp\{i:x_i=j\}$. Instead, we consider kernel estimates indexed by $0\leq c\leq1/(m-1)$
\[
\tilde{P}_j=\{ 1-(m-1)c\}\hat{P}_j+c\sum_{k\neq j}\hat{P}_k,\ j=1,\ldots,m.
\]
We use squared loss to compare these two estimators. After some calculations,
\begin{eqnarray*}
 &&E\{ L(\hat{P},P) \} =\sum_{j=1}^mE(\hat{P}_j-P_j)^2=\sum_{j=1}^m\frac{P_j(1-P_j)}{n_j},
\end{eqnarray*}
and $E\{ L(\tilde{P},P)\}=\sum_{j=1}^mE(\tilde{P}_j-P_j)^2$ is a function of $c$
obtaining minimum
\begin{eqnarray}
\lefteqn{ E\{ L(\hat{P},P)\} \bigg[ 1-\bigg(1-\frac{1}{m}\bigg)\frac{E\{ L(\hat{P},P) \} }{E\{ L(\hat{P},P)\}
+\frac{1}{m-1}\sum_{i<j}(P_i-P_j)^2}\bigg] } \nonumber \\
& & \qquad  \qquad \qquad \qquad \in\bigg(\frac{1}{m}E\{ L(\hat{P},P)\},E\{ L(\hat{P},P)\} \bigg), \nonumber
\end{eqnarray}
at
\[
c_0=\frac{1}{m}\frac{E\{ L(\hat{P},P) \} }{E\{ L(\hat{P},P)\}+\frac{1}{m-1}\sum_{i<j}(P_i-P_j)^2}\in\bigg(0,\frac{1}{m-1}\bigg).
\]
This suggests that when $P_j$'s are similar, the estimate $\tilde{P}$ can reduce the risk up to only $1/m$ the risk of estimating $\hat{P}$ separately. If $P_j$'s are not similar, $\tilde{P}$ can still reduce the risk considerably when the cell counts $\{ n_j \}$ are small.

\subsection{Posterior convergence rates}

Suppose we obtain data for $n$ observations $y^n = (y_1,\ldots,y_n)'$, which are conditionally independent given $X^n = (x_1,\ldots,x_n)'$ with
$x_i = (x_{i1},\ldots,x_{ip_n})'$, $x_{ij} \in \{1,\ldots,d\}$ and $p_n\gg n$.  We exclude the $n$ subscript on $p$ when convenient and assume $d_j=d$ for $j=1,\ldots,p$ for simplicity in exposition, though the results generalize directly.  Let $P_0$ denote the true data generating model, which can be dependent on $n$.  Rather than assume that most of the predictors have no impact on $Y$, we consider the situation similar to \cite{Jiang2006} that most have nonzero but very small influence. Specifically, parameterizing the true model $P_0$ in our tensor form with $k_j=d$ for $j=1,\ldots,p_n$ (this is always possible for any $P_0$), we assume:
\begin{quote}
{\bf Assumption A}. $\sum_{j=1}^{p_n}\max_{x_j}\sum_{h_j=2}^d\pi_{h_j}^{(j)}(x_j)<\infty$.
\end{quote}
This is a near sparsity restriction on $P_0$.  We additionally assume that the true conditional probabilities are strictly greater than zero,
\begin{quote}
{\bf Assumption B}. $P_0(y|x)\geq \epsilon_0$ for any $x,y$ for some $\epsilon_0>0$.
\end{quote}
The next theorem states the posterior contraction rate under our prior (\ref{eq:prior})-(\ref{eq:priork}).  We use  $f\prec g$ to mean that $f$ is less than $g$ up to a constant independent of $n$.  Recall that $r_n, \overline{r}_n$ are hyperparameters in the prior corresponding to the expected and maximum number of important predictors, respectively.
\begin{theorem}\label{thm:3}
Assume the design points $x_1,\ldots,x_n$ are independent observations from an unknown probability distribution $G_n$ on $\{1,\ldots,d\}^{p_n}$. Moreover, assume the prior is specified as in (\ref{eq:prior})-(\ref{eq:priork}) and Assumptions A and B hold.  Let $\epsilon_n$ be a sequence with $\epsilon_n\rightarrow0$, $n\epsilon_n^2\rightarrow\infty$ and $\sum_n\exp(-n\epsilon_n^2)<\infty$. Assume the following conditions hold: (i)
$\bar{r}_n\log p_n\prec n\epsilon_n^2$, (ii) $\bar{r}_nd^{\bar{r}_n}\log( \bar{r}_n / \epsilon_n) \prec n\epsilon_n^2$, (iii) $r_n/p_n\rightarrow0$ as $n \rightarrow\infty$, and (iv) there exists a sequence of models $\gamma_n$ with size $\bar{r}_n$ such that $\sum_{j\notin\gamma_n}\max_{x_j}\sum_{h_j=2}^d\pi_{h_j}^{(j)}(x_j)\prec\epsilon_n^2.$
Denote $d(P,P_0)=\int\sum_{y=1}^{d_0}\big|P(y|x_1,\ldots,x_p)-P_0(y|x_1,\ldots,x_p)\big|G_n(dx_1,\ldots,dx_p)$, then
\[
\Pi_n\big\{ P:d(P,P_0)\geq M\epsilon_n|y^n, X^n \big\}\rightarrow0\ a.s. P_0^n.
\]
\end{theorem}

The following corollary tells us that the posterior convergence rate of our model can be very close to $n^{-1/2}$ for appropriate hyperparameter choices.
\begin{corollary}
For any $\alpha\in(0,1)$, $\epsilon_n=n^{-(1-\alpha)/2}\log n$ will satisfy the conditions in Theorem \ref{thm:3} if
$r_n\prec\bar{r}_n\prec\log n$, $p_n\prec\exp(n^{\alpha})$ and there exists a sequence of models $\gamma_n$ with size $\bar{r}_n$ such that $\sum_{j\notin\gamma_n}\max_{x_j}\sum_{h_j=2}^d\pi_{h_j}^{(j)}(x_j)\prec n^{\alpha-1}\log^2n$.
\end{corollary}

The condition $\bar{r}_n\prec\log n$ is equivalent to $n/d^{\bar{r}_n}\succ1$, which means that to obtain good approximations to the true model, we cannot include too many predictors, and should ensure that on average there is order one observation in each cell. Based on the above observations, we recommend using $r_n=\log_d(n),\bar{r}_n=2r_n$ as default values for the prior in applications.

\section{Posterior Computation}
In section 4.1, we consider fixed $k=(k_1,\ldots,k_p)'$ and use a Gibbs sampler to draw posterior samples.  Generalizing this Gibbs sampler, we developed a
reversible jump Markov Chain Monte Carlo (RJMCMC) algorithm \citep{Green1995} to draw posterior samples from the joint distribution of $k=\{k_j:j=1,\ldots,p\}$ and $(\Lambda,\pi,z)$.  However, for $n$ and $p$ equal to several hundred or more, we were unable to design an RJMCMC algorithm that was sufficiently efficient to be used routinely. Hence, in section 4.2, we propose a faster two stage procedure based on approximated marginal likelihood.

\subsection{Gibbs sampling for fixed $k$}
Under (\ref{eq:prior}) the full conditional posterior distributions of $\Lambda$, $\pi$ and $z$ all have simple forms, which we sample from as follows.
\begin{enumerate}\itemsep=-3pt
  \item
  For $h_j=1,\ldots,k_j,j=1,\ldots,p$, update $\lambda_{h_1,\ldots,h_p}$ from the Dirichlet conditional, $\big\{ \lambda_{h_1,\ldots,h_p}(1),\ldots,\lambda_{h_1,\ldots,h_p}(d) \big\} |-\sim$
\begin{eqnarray*}
&&  \text{Diri}\bigg( \frac{1}{d} + \sum_{i=1}^n1(z_{i1}=h_1,\ldots,z_{ip}=h_p,y_i=1), \\
&& \qquad \ \ldots, \frac{1}{d} + \sum_{i=1}^n1(z_{i1}=h_1,\ldots,z_{ip}=h_p,y_i=d) \bigg).
\end{eqnarray*}
  \item
  Update $\pi^{(j)}(k)$ from the Dirichlet full conditional posterior distribution,
\begin{eqnarray*}
 \big\{ \pi_{1}^{(j)}(k),\ldots,\pi_{k_j}^{(j)}(k)\big\} |- & \sim & \text{Diri}\bigg(\frac{1}{k_j}+\sum_{i=1}^n1(z_{ij}=1)1(x_{ij}=k),\\
&& \qquad \  \  \ldots,\frac{1}{k_j}+\sum_{i=1}^n1(z_{ij}=k_j)1(x_{ij}=k)\bigg).
\end{eqnarray*}
  \item
  Update $z_{ij}$ from the multinomial full conditional posterior, with
  \[
  P(z_{ij}=h|-)\propto \pi_{h}^{(j)}(x_{ij})\lambda_{z_{i,1},\ldots,z_{i,j-1},h,z_{i,j+1},\ldots,z_{i,p}}(y_i).
  \]
\end{enumerate}

\subsection{Two step approximation}
We propose a two stage algorithm, which identifies a good model in the first stage and then learns the posterior distribution for this model in a second stage via the Gibbs sampler of section 4.1.  We first propose an approximation to the marginal likelihood. For simplicity in exposition, we focus on binary $Y$ with $d_0=2$, but the approach generalizes in a straightforward manner, with the beta functions in the below expression for the marginal likelihood replaced with functions of the form $\Gamma(a_1)\Gamma(a_2)\cdots\Gamma(a_{d_0})/\Gamma(a_1+\cdots+a_{d_0})$.  To motivate our approach, we first note that $\pi_{h_j}^{(j)}(x_j)$ can be viewed as providing a type of {\em soft} clustering of the $j$th feature $X_j$, controlling borrowing of information among probabilities conditional on combinations of predictors. To obtain approximated marginal likelihoods to be used only in the initial model selection stage, we propose to force $\pi_{h_j}^{(j)}(x_j)$ to be either zero or one, corresponding to a hard clustering of the predictors.  Under this approximation, the marginal likelihood has a simple expression.

 For a given model indexed by $k=\{k_j, j=1,\ldots,p\}$, we assume that the levels of $X_j$ are clustered into $k_j$ groups $A_{1}^{(j)},\ldots,A_{k_j}^{(j)}$. For example, with levels $\{1,2,3,4,5\}$, $A_1^{(j)}=\{1,2,3\}$ and $A_2^{(j)}=\{4,5\}$.
Then it is easy to see that the marginal likelihood conditional on $k$ and $A$ is $\mathcal{L}(y|k,A) = $
\begin{eqnarray*}
\lefteqn{ \prod_{h_1,\ldots,h_p}\frac{1}{\text{Beta}(1/2,1/2)}
\text{Beta}\bigg( \frac{1}{2} + \sum_{i=1}^n I(x_{i1}\in A_{h_1}^{(1)},\ldots,x_{ip}\in A_{h_p}^{(p)},y_i=1),} \\
&& \qquad \qquad \qquad \qquad \qquad \ \frac{1}{2} + \sum_{i=1}^n I(x_{i1}\in A_{h_1}^{(1)},\ldots,x_{ip}\in A_{h_p}^{(p)},y_i=0)\bigg).
\end{eqnarray*}
Having an expression for the marginal likelihood, we apply a stochastic search MCMC algorithm  \citep{George1997} to obtain samples of $(k_1,\ldots,k_p)$ from the approximated posterior distribution.  This proceeds as follows.
\begin{enumerate}\itemsep=-3pt
  \item For $j=1$ to $p$, do the following. Given the current model indexed by $k=\{k_j:j=1,\ldots,p\}$ and clusters $A=\{A_{h}^{(j)}:h=1,\ldots,k_j,j=1,\ldots,p\}$, propose to
      increase $k_j$ to $k_j+1$ (if $k_j<d$) or reduce it to $k_j-1$ (if $k_j>1$) with equal probability.
  \item If increase, randomly split a cluster of $X_j$ into two clusters (all splits have equal probability).
      For example, if $d_{j}=5$, $k_{j}=2$ and the levels of $X_{j}$ are clustered as $\{1,2,3\}$ and $\{4,5\}$. There are $4$ possible splitting schemes: three ways to split $\{1,2,3\}$ and one way to split $\{4,5\}$. We randomly choose one.  Accept this move with acceptance rate based on the approximated marginal likelihood.
  \item If decrease, randomly merge two clusters and accept or reject this move.
\end{enumerate}
Estimating approximated marginal inclusion probabilities of $k_j>1$ based on this algorithm, we keep predictors having inclusion probabilities great than 0.5; this leads to selecting the median probability model, which in simpler settings has been shown to have optimality properties in terms of predictive performance \citep{Barbieri2004}.

\section{Simulation Studies}
To assess the performance of the proposed approach, we conducted a simulation study and calculated the misclassification rate on the testing samples. Simulated data consisted of $N=2,000$ instances with $p=600$ covariates $X_1,\ldots,X_p$, each of which has $d=4$ levels, and a binary response $Y$. We
assumed that the true model had three important predictors $X_9,X_{11}$ and $X_{13}$, and generated $P(Y=1|X_9=x_9,X_{11}=x_{11},X_{13}=x_{13})$ independently for each combination of $(x_9,x_{11},x_{13})$.  To obtain an average Bayes error rate (optimal misclassification rate) around 15\%, we generated the conditional probabilities from $f(U)=U^2/\{U^2+(1-U)^2\}$, where $U\sim \text{Unif}(0,1)$. Each time, we randomly chose $n$ samples as training with the remaining $N-n$ as testing.  We implemented our approach using the training set and calculated the test sample misclassification rate corresponding to the average MSE defined as
$$
\text{aMSE}=\frac{1}{4^p}\sum_{x_1,\ldots,x_p}\big\{ P(Y=1|x_1,\ldots,x_p)-\hat{P}(Y=1|x_1,\ldots,x_p)\big\}^2,
$$
where $\hat{P}$ is the fitted conditional probability.
We selected four training sizes $n=200$, $400$, $600$ and $800$. For each training size, we randomly chose 10 training-test splits and used our two stage algorithm to fit the model for each split. According to our theoretical results, we chose $r=\log_4(n)$ as the expected number of important predictors in the prior. We ran 1,000 iterations for the first stage and 2,000 iterations for the second stage, treating the first half as burn-in. In addition, we compared the results with the random forests algorithm \citep{Breiman2001} applied to the same training-test split data.

Table 1 displays the results.  In the very challenging case in which the training sample size was only 200, both methods had poor performance.  However, as the training sample size increased, the proposed conditional tensor factorization method rapidly approached the optimal 15\%, with excellent performance even in the $n=p=600$ case.  In contrast, random forests had consistently poor performance in this challenging setting involving a low signal strength, a modest sample size, and moderately large numbers of candidate predictors.  In addition to the clearly superior classification performance, our method had the advantage of providing variable selection results.  Table 2 provides the average approximated marginal inclusion probabilities for the three important predictors and remaining predictors for each training sample size.  Consistently with the results in Table 1, the method fails to detect the important predictors when the training sample size is only $n=200$ but as the sample size increases appropriately assigns high marginal inclusion probabilities to the important predictors and low ones to the unimportant predictors.   This is just one initial set of simulations against one competitor, which is often thought to provide good performance in classification problems, but the results are promising.

\begin{table}
    \caption{Testing Results for Synthetic Data Example. RF: random forests; TF: Our tensor factorization model. }
\centering
\fbox{%
    \begin{tabular}{c|cccc}
      training size & 200 & 400 & 600 & 800 \\
      \hline
      aMSE & 0.144 & 0.042 & 0.024 & 0.010 \\
      Misclassification Rate of TF & 0.503 & 0.288 & 0.189 & 0.168 \\
      Misclassification Rate of RF & 0.496 & 0.482 & 0.471 & 0.472\\
    \end{tabular}}
    \label{ta:1}
\end{table}

\begin{table}
    \caption{Variable Selection Results for Synthetic Data Example. Columns 2-4 are approximated inclusion probabilities of the $9$th,$11$th,$13$th predictors. Column 5 is the maximum inclusion probability across the remaining predictors. Column 6 is the average inclusion probability across the remaining predictors. These quantities are averages over 10 trials.}
\centering
\fbox{%
    \begin{tabular}{c|ccccc}
      training size & 9 & 11 & 13
      & Max & Average\\
            \hline
      200 & 0.092 & 0.041 & 0.063 & 0.161 & 0.002\\
      400 & 0.816 & 0.820 & 0.808 & 0.013 & 0.000\\
      600 & 1.000 & 1.000 & 1.000 & 0.000 & 0.000\\
      800 & 1.000 & 1.000 & 1.000 & 0.000 & 0.000\\
    \end{tabular}}
    \label{ta:2}
\end{table}

\section{Applications}
We compare our method with other competing methods in three data sets from the UCI repository. The first data set is Promoter Gene Sequences (abbreviated as promoter data below). The data consists of A, C, G, T nucleotides at $p = 57$ positions for $N = 106$ sequences and a binary response indicating instances of promoters and non-promoters. We use 5-fold cross validation with $n=85$ training samples and $N-n=21$ test samples in each training-test split.

The second data set is the Splice-junction Gene Sequences (abbreviated as splice data below). These data consist of A, C, G, T nucleotides at $p = 60$ positions for $N = 3,175$ sequences. Each sequence belongs to one of the three classes: exon/intron boundary (EI), intron/exon boundary (IE) or neither (N). Since its sample size is much larger than the first data set, we compare our approach with competing methods in two scenarios: a small sample size and a moderate sample size. In the small sample size case, each time we randomly select $n=200$ instances as training and calculate the misclassification rate on the testing set composed of the remaining $2,975$ instances. We repeat this for each method for five training-test splits and report the average misclassification rate. In the moderate sample size case, we use 5-fold cross validation so that each time $n=2,540$ instances are treated as training data.

The third data set describes diagnosing of cardiac Single Proton Emission Computed Tomography (SPECT) images. Each of the patients is classified into two categories: normal and abnormal. The database of 267 SPECT image sets (patients) has 22 binary feature patterns. This data set has been previously divided into a training set of size 80 and a testing set of size 187.

As competitors we considered lasso penalized logistic regression \citep{Park2007} and 5 black-box algorithms: CART, random forests \citep{Breiman2001}, neural networks with two layers of hidden units, support vector machines and Bayesian additive regression trees \citep{Chipman2006}. Among them, BART was not implemented in the splice data since we were unable to find a multi-class implementation of their approach.

\begin{table}
 \caption{UCI Data Example. RF: random forests, NN: neural networks, SVM: support vector machine, BART: Bayesian additive regression trees, TF: Our tensor factorization model. Misclassification rates are displayed.}
\centering \fbox{%
\begin{tabular}{cccccccc}
Data       & CART & RF & NN & LASSO & SVM & BART & TF \\
      \hline
      Promoter (n=85) & $0.236$ & $0.066$ & $0.170$ & $0.075$ & $0.151$ & $0.113$ & $0.066$ \\
      Splice (n=200) & 0.161 & 0.122 & 0.226 & 0.141 & 0.286 & - & 0.112 \\
      Splice (n=2540) & 0.059 & 0.046 & 0.165 & 0.123 & 0.059 & - & 0.058\\
      SPECT (n=80) & 0.312 & 0.235 & 0.278 & 0.277 & 0.246 & 0.225 & 0.198\\
    \end{tabular}}
        \label{ta:3}
\end{table}

Table 3 shows the results. Our method produced at worst comparable classification accuracy to the best of the competitors in each of the cases considered.  Consistent with our previous experience in more compressive comparisons of classifiers, Random Forests (RF) provided the best competitor overall, justifying our focus on RF in the simulation examples above.  We expect our approach to do particularly well when there is a modest training sample size and high-dimensional predictors.  We additionally have an advantage in terms of interpretability over several of these approaches, including RF and BART, in conducting variable selection.  For example, in the promoter data, our model selected nucleotides at $15$th, $16$th, $17$th, and $39$th positions as important predictors when the full data are used to fit the model. In the splice data, the $28$th, $29$th, $30$th, $31$st, $32$nd and $35$th positions are selected. These results are reasonable since for nucleotide sequences, nearby nucleotides form a motif regulating important functions.  In the SPECT data, the $11$st, $13$rd and $16$th predictors are selected.  It is notable that in each of these cases we obtained excellent classification performance based on a small subset of the predictors.

\section{Discussion}
This article proposes a framework for nonparametric Bayesian classification relying on a novel class of conditional tensor factorizations.  The nonparametric Bayes framework is appealing in facilitating variable selection and uncertainty about the core tensor dimensions in the Tucker-type factorization, while avoiding the need for parameter tuning.  In particular, we have recommended a single default prior setting that can be used in general applications without relying on cross-validation or other approaches for estimating tuning parameters.  One of our major contributions is the strong theoretical support we provide for our proposed approach.  Although it has been commonly observed that Bayesian parametric and nonparametric methods have practical gains in numerous applications, there is a clear lack of theory supporting these empirical gains.

Interesting ongoing directions include developing faster approximation algorithms and generalizing the conditional tensor factorization model to accommodate broader feature modalities.  In the fast algorithms direction, online variational methods \citep{Hoffman2010} provide a promising direction.  Regarding generalizations, we can potentially accommodate continuous predictors and more complex {\em object} predictors (text, images, curves, etc) through probabilistic clustering of the predictors in a first stage, with $X_j$ then corresponding to the cluster index for feature $j$.

\section*{Acknowledgments}
This research was supported by grant ES017436 from the National Institute of Environmental Health Sciences (NIEHS) of the National Institutes of Health
(NIH).

\section*{Appendix A: Proof of Theorem \ref{thm:1}}

\begin{proof}[Theorem \ref{thm:1}]
First reshape $P(y|x_1,\ldots,x_p)$ according to $x_1$ as a matrix $A^{(1)}$ of size $d_1\times d_0d_2d_3\ldots d_p$, with the $h^{th}$ row a long vector,
\begin{eqnarray*}
\lefteqn{ \big\{ P(1|h,1,\ldots,1,1),P(1|h,1,\ldots,1,2),\ldots,P(1|h,1,\ldots,1,d_p), } \\
& &  P(1|h,1,\ldots,2,1),\ldots,P(1|h,1,\ldots,2,d_j),\ldots,P(d_0|h,d_2,\ldots,d_{p-1},d_p)\big\},
 \end{eqnarray*}
denoted $A^{(1)}\{ h,(y,x_2,\ldots,x_p)\}$. Let $k_1$ be the smallest number such that
\begin{equation}\label{eq:5}
P(y|x_1,\ldots,x_p)=A^{(1)}\{ x_1,(y,x_2,\ldots,x_p)\}=\sum_{h=1}^{k_1}\lambda^{(1)}_{hx_2\ldots x_p}(y)\pi_{h}^{(1)}(x_1),
\end{equation}
subject to $\sum_{y=1}^{d_0}\lambda^{(1)}_{hx_2\ldots x_p}(y) = 1$ for each $(h,x_2,\ldots,x_p)$,
        $\sum_{h=1}^{k_1}\pi_{h}^{(1)}(x_1)=1$ for each $x_1$,
        $\lambda^{(1)}_{h x_2\ldots x_p}(y)\geq0$, and $\pi_{h}^{(1)}(x_1)\geq0.$
In this non-negative matrix factorization, $k_1 \le d_1$ exists because
        $\lambda^{(1)}_{hx_2\ldots x_p}(y)=P(y|h,x_2,\ldots,x_p),$ $\pi_{h}^{(1)}(x_1)=\delta_{hx_1}$ with $k_j=d_j$ is a choice satisfying the above equations. Here $\delta_{ij}=1$ if $i=j$, $\delta_{ij}=0$ if $i\neq j$ is Kronecker delta function.

Taking $\lambda^{(1)}_{x_1x_2\ldots x_p}(y)$ from (\ref{eq:5}) with argument $x_2$, we can apply the same type of decomposition to obtain
\[
\lambda^{(1)}_{x_1x_2\ldots x_p}(y)=\sum_{h=1}^{k_2}\lambda^{(2)}_{x_1hx_3\ldots x_p}(y)\pi_{h}^{(2)}(x_2),
\]
subject to $\sum_{y=1}^{d_0}\lambda^{(2)}_{x_1h\ldots x_p}(y) = 1$, for each $(x_1,h,\ldots,x_p)$, $\sum_{h=1}^{k_2}\pi_{h}^{(2)}(x_2)= 1$, for each $x_2$,
 $\lambda^{(2)}_{x_1h\ldots x_p}(c) \geq 0$, and $\pi_{h}^{(2)}(x_2)\geq0.$  Plugging back into equation (\ref{eq:5}),
\[
P(y|x_1,\ldots,x_p)= \sum_{h_1=1}^{k_1}\sum_{h_2=1}^{k_2}\lambda^{(2)}_{h_1h_2x_3\ldots x_p}(y)\pi_{h_1}^{(1)}(x_1)\pi_{h_2}^{(2)}(x_2).
\]
Repeating this procedure another $(p-2)$ times, we obtain equation (\ref{eq:1}) with $\lambda_{h_1h_2\ldots h_p}(y)=\lambda^{(p)}_{h_1h_2\ldots h_p}(y)$ and constraints (\ref{eq:2}).

\end{proof}

\section*{Appendix B: Proof of Theorem \ref{thm:3}}

To prove Theorem \ref{thm:3} we need some preliminaries. The following theorem is a minor modification of Theorem 2.1 in \cite{Ghosal2000} and the proof is included in a supplemental appendix.  For simplicity in notation, we denote the observed data for subject $i$ as $X_i$ with $X_i \stackrel{iid}{\sim} P \in \mathcal{P}$, $P \sim \Pi$, and the true model $P_0$.
\begin{theorem}\label{thm:2}
Let $\epsilon_n$ be a sequence with $\epsilon_n\rightarrow0$, $n\epsilon_n^2\rightarrow\infty$, $\sum_n\exp(-n\epsilon_n^2)<\infty$. Let $d$ be the total variance distance, $C>0$ be a constant and sets $\mathcal{P}_n\subset\mathcal{P}$. Define the following conditions:
\begin{enumerate}
  \item $logN(\epsilon_n,\mathcal{P}_n,d)\leq n\epsilon_n^2$;
  \item $\Pi_n(\mathcal{P}\backslash\mathcal{P}_n)\leq\exp\{-(2+C)n\epsilon_n^2\}$;
  \item $\Pi_n(P:||\log\frac{P}{P_0}||_{\infty}<\epsilon_n^2)>\exp(-Cn\epsilon^2_n)$.
\end{enumerate}
If the above conditions hold for all $n$ large enough, then for $M$ sufficiently large,
\[
\Pi_n\{ P:d(P,P_0)\geq M\epsilon_n|X_1,\ldots,X_n\}\rightarrow0\ a.s.P_0^n.
\]
\end{theorem}
In our case, $X_i$ include the response $y_i$ and predictors $x_i$, $P$ is the random measure characterizing the unknown joint distribution of $(y_i,x_i)$ and $P_0$ is the measure characterizing the true joint distribution.  As our focus is on the conditional probability, $P(y|x)$, we fix the marginal distribution of $X$ at it's true value $P_0(x)$ and model the unknown conditional $P(y|x)$ independently of the marginal of $X$.  By doing so, it is straightforward to show that we can ignore the marginal of $X$ in using Theorem 2 to study posterior convergence.  We simply restrict $\mathcal{P}$ to the set of joint probabilities such that $P(x) \equiv P_0(x)$.  The total variation distance between the joint probabilities $P$ and $P_0$ is equivalent to the distance between the conditionals defined in Theorem 2 by the identity
\begin{eqnarray*}
\lefteqn{ \int\sum_{y=1}^{d_0}\big|P(y,x_1,\ldots,x_p)-P_0(y,x_1,\ldots,x_p)\big|dx_1\cdots dx_p=}\\
&&\int\sum_{y=1}^{d_0}\big|P(y|x_1,\ldots,x_p)-P_0(y|x_1,\ldots,x_p)\big|dG_n(dx_1,\cdots, dx_p).
\end{eqnarray*}
Therefore, we will not distinguish the joint probability and the conditional probability and use $P$ to denote both of them henceforth.

To prove Theorem 2, we also need upper bounds on the distance between two models specified by (\ref{eq:1}) when the models are the same size and when they are nested.
\begin{lemma}\label{le:2}
Let $P$ and $\tilde{P}$ be two models specified by (\ref{eq:2}) with parameter $(\lambda,\pi)$ and $(\tilde{\lambda},\tilde{\pi})$, respectively. Then
\[
d(P,\tilde{P})\leq\sum_{y=1}^{d_0}\max_{h_1,\ldots,h_p}|\lambda_{h_1h_2\ldots h_p}(y)-\tilde{\lambda}_{h_1h_2\ldots h_p}(y)|+d_0\sum_{j=1}^p\max_{x_j,h_j}|\pi_{h_j}^{(j)}(x_j)-\tilde{\pi}_{h_j}^{(j)}(x_j)|.
\]
\end{lemma}

\begin{proof}[Lemma \ref{le:2}]
By definition of $d(P,\tilde{P})$, we only need to prove that for any $y=1,\ldots,d_o$ and any combination of $(x_1,\ldots,x_p)$,
\begin{eqnarray}
|P(y|x_1,\ldots,x_p)-\tilde{P}(y|x_1,\ldots,x_p)|&\leq&\max_{h_1,\ldots,h_p}|\lambda_{h_1h_2\ldots h_p}(y)-\tilde{\lambda}_{h_1h_2\ldots h_p}(y)| \nonumber \\
&&+\sum_{j=1}^p\max_{h_j}|\pi_{h_j}^{(j)}(x_j)-\tilde{\pi}_{h_j}^{(j)}(x_j)|.\label{eq:7}
\end{eqnarray}
Actually,
\begin{equation*}
|P(y|x_1,\ldots,x_p)-\tilde{P}(y|x_1,\ldots,x_p)|\leq A+\sum_{s=1}^pB_s,
\end{equation*}
where
\begin{equation*}
    \begin{split}
      &A=\sum_{h_1=1}^{k_1}\cdots\sum_{h_p=1}^{k_p}|\lambda_{h_1h_2\ldots h_p}(y)-\tilde{\lambda}_{h_1h_2\ldots h_p}(y)|\prod_{j=1}^p\pi_{h_j}^{(j)}(x_j) \nonumber \\
      &\leq\max_{h_1,\ldots,h_p}|\lambda_{h_1h_2\ldots h_p}(y)-\tilde{\lambda}_{h_1h_2\ldots h_p}(y)|\sum_{h_1=1}^{k_1}\cdots\sum_{h_p=1}^{k_p}\prod_{j=1}^p\pi_{h_j}^{(j)}(x_j) \\
      &=\max_{h_1,\ldots,h_p}|\lambda_{h_1h_2\ldots h_p}(y)-\tilde{\lambda}_{h_1h_2\ldots h_p}(y)|,
     \end{split}
\end{equation*}
where the last step is by using the second equation in (\ref{eq:3}), and
\begin{equation*}
    \begin{split}
&B_s=\sum_{h_1=1}^{k_1}\cdots\sum_{h_p=1}^{k_p}\tilde{\lambda}_{h_1h_2\ldots h_p}(y)\ |\pi_{h_s}^{(s)}(x_s)-\pi_{h_s}^{(s)}(x_s)|\ \prod_{j=1}^{s-1}\tilde{\pi}_{h_j}^{(j)}(x_j)
\prod_{j=s+1}^{p}\pi_{h_j}^{(j)}(x_j)\\
&\leq\max_{h_s}|\pi_{h_s}^{(s)}(x_j)-\tilde{\pi}_{h_s}^{(s)}(x_j)|,
     \end{split}
\end{equation*}
where the last step is again by using the second equation in (\ref{eq:2}) and the fact that $\lambda_{h_1h_2\ldots h_p}(y)\leq1$.
Combining the above inequalities we can obtain (\ref{eq:7}).
\end{proof}

\begin{lemma}\label{le:3}
Let $P$ and $\tilde{P}$ be two models as in (\ref{eq:2}) with parameters $(\lambda,\pi)$ and $(\tilde{\lambda},\tilde{\pi})$, respectively. Suppose $P$ is nested in $\tilde{P}$, i.e. there exists a number $r$ s.t.
\[
\lambda_{h_1\cdots h_r h_{r+1}\cdots h_p}=\tilde{\lambda}_{h_1\cdots h_r1\cdots1}, \text{ for } h_j\leq d,j=1,\ldots,p,
\]
\[
\pi^{(j)}_{h_j}(x_j)=\tilde{\pi}^{(j)}_{h_j}(x_j),j\leq r, \ \ \pi^{(j)}_{h_j}(x_j)=I(h_j=1),j>r,
\]
Then
\[
d(P,\tilde{P})\leq d_0\sum_{j=r+1}^p\max_{x_j}\sum_{h_j=2}^d\tilde{\pi}_{h_j}^{(j)}(x_j).
\]
\end{lemma}

\begin{proof}[Lemma \ref{le:3}]
By the constraints on $\pi$ in (\ref{eq:2}),
\begin{eqnarray*}
\lefteqn{ |P(y|x_1,\ldots,x_p)-\tilde{P}(y|x_1,\ldots,x_p)|  } \nonumber \\
&\leq& \sum_{h_{r+1}=1}^d\cdots\sum_{h_p=1}^d
\max_{h_1,\ldots,h_r}|\tilde{\lambda}_{h_1\cdots h_{r}1\cdots1}(y)-\tilde{\lambda}_{h_1\ldots h_p}(y)| \prod_{j=r+1}^p\tilde{\pi}_{h_j}^{(j)}(x_j)\\
&\leq& \sum_{h_{r+1}=2}^d\cdots\sum_{h_p=1}^d
\max_{h_1,\ldots,h_r}|\tilde{\lambda}_{h_1\cdots h_{r}1\cdots1}(y)-\tilde{\lambda}_{h_1\ldots h_p}(y)| \prod_{j=r+1}^p\tilde{\pi}_{h_j}^{(j)}(x_j)\\
&&+\cdots+\sum_{h_{r+1}=1}^d\cdots\sum_{h_p=2}^d \max_{h_1,\ldots,h_r}|\tilde{\lambda}_{h_1\cdots h_{r}1\cdots1}(y)-\tilde{\lambda}_{h_1\ldots h_p}(y)| \prod_{j=r+1}^p\tilde{\pi}_{h_j}^{(j)}(x_j).
\end{eqnarray*}
The lemma can be proved by noticing $\tilde{\lambda}_{h_1\ldots h_p}(y)\in[0,1]$.
\end{proof}

\begin{proof}[Theorem \ref{thm:3}]
We verify conditions (a)-(c) in Theorem \ref{thm:2}. As we described previously, we do not need to distinguish the joint probability and the conditional probability under our prior specification. Let $\mathcal{P}_n$ be all conditional probability tensors having no more than $\bar{r}_n$ predictors, so that
$|\gamma_n|\leq\bar{r}_n$.

\textit{Condition (a):}
By the conclusion of lemma \ref{le:2}, we know that an $\epsilon_n$-net $E_n$ of $\mathcal{P}_n$ can be chosen so that for each $(\gamma,\lambda,\pi)\in\mathcal{P}_n$ that satisfies constraints (\ref{eq:2}), there exists $(\tilde{\gamma},\tilde{\lambda},\tilde{\pi})\in E_n$ such that $\tilde{\gamma}=\gamma$, $\max_{y,h_1,\ldots,h_p}|\lambda_{h_1h_2\ldots h_p}(y)-\tilde{\lambda}_{h_1h_2\ldots h_p}(y)|<\frac{\epsilon_n}{(\bar{r}_n+1)d_0}$ and $\max_{x_j,h_j}|\pi_{h_j}^{(j)}(x_j)-\tilde{\pi}_{h_j}^{(j)}(x_j)|<\frac{\epsilon_n}{(\bar{r}_n+1)d_0}$ for $j\in\gamma$. Hence, we can pick $d$-balls of the form
 $$\gamma\times\prod_{h_1,\ldots,h_p,y}\bigg(\lambda_{h_1h_2\ldots h_p}(y)\pm\frac{\epsilon_n}{(\bar{r}_n+1)d_0}\bigg)\times\prod_{j=1}^{p_n}
 \prod_{h_j,x_j}\bigg(\pi_{h_j}^{(j)}(x_j)\pm\frac{\epsilon_n}{(\bar{r}_n+1)d_0}\bigg).$$
For each fixed size model with $|\gamma|\leq\bar{r}_n$ in $\mathcal{P}_n$, there are at most $d_0d^{\bar{r}_n}$ $\lambda_{h_1h_2\ldots h_p}(y)$'s and $\bar{r}_n d^2$ $\pi_{h_j}^{(j)}(x_j)$'s. For $\gamma$, there are at most $p_n^r$ models of size $r$. Hence, the log of the minimal  number of size-$\epsilon_n$ balls needed to cover $\mathcal{P}_n$ is at most
\[
\log\big\{ (\bar{r}_n+1)p_n^{\bar{r}_n}\big\} +{\bar{r}_n}d_0d^{\bar{r}_n+2}\log\frac{(\bar{r}_n+1)d_0}{2\epsilon_n}.
\]
By the conditions in the theorem, each term is bounded by some constant  $ \times n\epsilon_n^2$, and we can adjust these constants to make it less than $n\epsilon_n^2$.

\textit{Condition (b):}
Because $\Pi_n(\mathcal{P}_n^c)=0$ in our case, this condition is trivially satisfied. Actually, this condition will still be satisfied as long as $\Pi_n(\gamma>\bar{r}_n)\leq\exp\{ -(2+C)n\epsilon_n^2\}$, which implies that the prior probability assigned to large models is exponentially small.

\textit{Condition (c):}
As $P_0$ is lower bounded away from zero by $\epsilon_0$, $||\log\frac{P}{P_0}||_{\infty}<\epsilon_n^2$ is implied by $||P-P_0||_{\infty}<\epsilon_0\epsilon_n^2$ for $n$ large enough ($\epsilon_n\rightarrow0$ as $n$ increases). Let $(\tilde{\lambda},\tilde{\pi})$ denote parameters for the true model $P_0$.
Applying lemma \ref{le:2} to bound $d(P,\bar{P})$, where $\bar{P}(y|x_1,\ldots,x_p)\equiv P_0(y|x_1,\ldots,x_{\bar{r}_n},1,\ldots,1)$, and then estimating the difference between $\bar{P}$ and $P_0$ by lemma \ref{le:3}, we have
\begin{equation}\label{eq:8}
\begin{split}
&d(P,P_0)\leq\max_{h_1,\ldots,h_{\bar{r}_n}}|\lambda_{h_1h_2\ldots h_{\bar{r}_n}}(y)-\tilde{\lambda}_{h_1h_2\ldots h_{\bar{r}_n}1\ldots1}(y)|\\
&+\sum_{j=1}^{\bar{r}_n}\max_{x_j,h_j}\big|\pi_{h_j}^{(j)}(x_j)-\tilde{\pi}_{h_j}^{(j)}(x_j)\big|
+\sum_{j=\bar{r}_n+1}^{p_n}\max_{x_j}\sum_{h_j=2}^d\tilde\pi_{h_j}^{(j)}(x_j).
\end{split}
\end{equation}
Combining (\ref{eq:8}) and condition (iv) in Theorem \ref{thm:3}, and assuming without loss of generality that $\gamma_n$ corresponds to the first $\bar{r}_n$ covariates,
$||\log\frac{P}{P_0}||_{\infty}<\epsilon_n^2$ is implied by
\begin{equation*}
    \begin{split}
&\max_{h_1,\ldots,h_{\bar{r}_n}}|\lambda_{h_1h_2\ldots h_{\bar{r}_n}}(y)-\tilde{\lambda}_{h_1h_2\ldots h_{\bar{r}_n}1\ldots1}(y)|\prec\frac{\epsilon_n^2}{\bar{r}_n+1},\\
&\max_{h_j}|\pi_{h_j}^{(j)}(x_j)-\tilde{\pi}_{h_j}^{(j)}(x_j)|\prec\frac{\epsilon_n^2}{\bar{r}_n+1}.
     \end{split}
\end{equation*}
Moreover, the Dir$(1/d,\ldots,1/d)$ and Dir$(1/d_0,\ldots,1/d_0)$ priors for $\lambda_{h_1h_2\ldots h_{\bar{r}_n}}(\cdot)$ and $\pi_{\cdot}^{(j)}(x_j)$ have density lower bounded away from zero by a constant not involving $n$, and the prior $P(\gamma=\gamma_n)$ is $(r_n/p_n)^{\bar{r}_n}(1-r_n/p_n)^{p_n-\bar{r}_n}$, so as $r_n/p_n\rightarrow0$, $\log\Pi_n(\gamma=\gamma_n)\sim\bar{r}_n\log(r_n/p_n)\geq-\bar{r}_n\log p_n$. Combining these and the conditions in the theorem,
\begin{equation*}
    \begin{split}
&\log\Pi_n\bigg(P:||\log\frac{P}{P_0}||_{\infty}<\epsilon_n^2 \bigg)\succ \bar{r}_nd_0d^{\bar{r}_n}\log\frac{\epsilon_n^2}{(\bar{r}_n+1)d_0}-\bar{r}_n\log p_n\succ -n\epsilon_n^2.
     \end{split}
\end{equation*}
By adjusting constants, we can make $\log\Pi_n(P:||\log\frac{P}{P_0}||_{\infty}<\epsilon_n^2)>-Cn\epsilon_n^2$.
\end{proof}

\bibliography{draft}

\begin{thebibliography}{}

\bibitem[\protect\citeauthoryear{Barbieri and Berger}{Barbieri and
  Berger}{2004}]{Barbieri2004}
Barbieri, M.~M. and J.~O. Berger (2004).
\newblock Optimal predictive model selection.
\newblock {\em Ann. Statist.\/}~{\em 32}, 870--897.

\bibitem[\protect\citeauthoryear{Bhattacharya and Dunson}{Bhattacharya and
  Dunson}{2011}]{Bhattacharya2011}
Bhattacharya, A. and D.~B. Dunson (2011).
\newblock Sparse {B}ayesian infinite factor models.
\newblock {\em Biometrika\/}~{\em 98}, 291--306.

\bibitem[\protect\citeauthoryear{Bhattacharya and Dunson}{Bhattacharya and
  Dunson}{2012}]{Anirban2011}
Bhattacharya, A. and D.~B. Dunson (2012).
\newblock Simplex factor models for multivariate unordered categorical data.
\newblock {\em J. Amer. Statist. Assoc.\/}~{\em 107}, 362--377.

\bibitem[\protect\citeauthoryear{Breiman}{Breiman}{2001}]{Breiman2001}
Breiman, L. (2001).
\newblock Random forests.
\newblock {\em Mach. Learn.\/}~{\em 45}, 5--32.

\bibitem[\protect\citeauthoryear{Breiman, Friedman, Olshen, and Stone}{Breiman
  et~al.}{1984}]{Breiman1984}
Breiman, L., J.~H. Friedman, R.~A. Olshen, and C.~J. Stone (1984).
\newblock {\em Classification and Regression Trees}.
\newblock Belmont, CA: Wadsworth, Inc.

\bibitem[\protect\citeauthoryear{Castillo and Van Der~Vaart}{Castillo and Van
  Der~Vaart}{2012}]{Castillo2012}
Castillo, I. and A.~Van Der~Vaart (2012).
\newblock Needles and straw in a haystack: Posterior concentration for possibly
  sparse sequences.
\newblock {\em Submitted to Ann. Statist.\/}.

\bibitem[\protect\citeauthoryear{Chipman, George, and E.}{Chipman
  et~al.}{2010}]{Chipman2006}
Chipman, H.~A., E.~I. George, and M.~R. E. (2010).
\newblock {BART}: {B}ayesian additive regression trees.
\newblock {\em Ann. Appl. Stat.\/}~{\em 4}, 266--298.

\bibitem[\protect\citeauthoryear{Cohen and Rothblum}{Cohen and
  Rothblum}{1993}]{Cohen1993}
Cohen, J.~E. and U.~G. Rothblum (1993).
\newblock Nonnegative ranks, decompositions, and factorizations of nonnegative
  matrices.
\newblock {\em Linear Algebra Appl.\/}~{\em 190}, 37.

\bibitem[\protect\citeauthoryear{De~Lathauwer, De~Moor, and
  Vanderwalle}{De~Lathauwer et~al.}{2000}]{DeLathauwer2000}
De~Lathauwer, L., B.~De~Moor, and J.~Vanderwalle (2000).
\newblock A multilinear singular value decomposition.
\newblock {\em SIAM J. Matrix Anal. Appl.\/}~{\em 21}, 1253--1278.

\bibitem[\protect\citeauthoryear{Dunson and Xing}{Dunson and
  Xing}{2009}]{Dunson2009}
Dunson, D.~B. and C.~Xing (2009).
\newblock Nonparametric {B}ayes modeling of multivariate categorical data.
\newblock {\em J. Amer. Statist. Assoc.\/}~{\em 104}, 1042--1051.

\bibitem[\protect\citeauthoryear{Genkin, Lewis, and Madigan}{Genkin
  et~al.}{2007}]{Alexander2007}
Genkin, A., D.~D. Lewis, and D.~Madigan (2007).
\newblock Large-scale {B}ayesian logistic regression for text categorization.
\newblock {\em Technometrics\/}~{\em 49}, 291--304.

\bibitem[\protect\citeauthoryear{George and McCulloch}{George and
  McCulloch}{1997}]{George1997}
George, E. and R.~McCulloch (1997).
\newblock Approaches for {B}ayesian variable selection.
\newblock {\em Statist. Sinica\/}~{\em 7}, 339--373.

\bibitem[\protect\citeauthoryear{Ghosal, Ghosh, and Van Der~Vaart}{Ghosal
  et~al.}{2000}]{Ghosal2000}
Ghosal, H., J.~K. Ghosh, and A.~W. Van Der~Vaart (2000).
\newblock Convergence rates of posterior distributions.
\newblock {\em Ann. Statist.\/}~{\em 28}, 500--531.

\bibitem[\protect\citeauthoryear{Ghosh and Dunson}{Ghosh and
  Dunson}{2009}]{Ghosh2009}
Ghosh, J. and D.~B. Dunson (2009).
\newblock Default prior distributions and efficient posterior computation in
  {B}ayesian factor analysis.
\newblock {\em J. Comput. Graph. Statist.\/}~{\em 18}, 306--320.

\bibitem[\protect\citeauthoryear{Green}{Green}{1995}]{Green1995}
Green, P. (1995).
\newblock Reversible jump {M}arkov chain {M}onte {C}arlo computation and
  {B}ayesian model determination.
\newblock {\em Biometrika\/}~{\em 82}, 711--732.

\bibitem[\protect\citeauthoryear{Harshman}{Harshman}{1970}]{Harshman1970}
Harshman, R. (1970).
\newblock Foundations of the {PARAFAC} procedure: Models and conditions for an
  `exploratory' multi-modal factor analysis.
\newblock {\em UCLA working papers in phonetics\/}~{\em 16}, 1--84.

\bibitem[\protect\citeauthoryear{Harshman and Lundy}{Harshman and
  Lundy}{1994}]{Harshman1994}
Harshman, R. and M.~Lundy (1994).
\newblock Parallel factor analysis.
\newblock {\em Comput. Statist. Data Anal.\/}~{\em 18}, 39--72.

\bibitem[\protect\citeauthoryear{Hoffman, Blei, and Bach}{Hoffman
  et~al.}{2010}]{Hoffman2010}
Hoffman, M., D.~Blei, and F.~Bach (2010).
\newblock Online learning for latent {D}irichlet allocation.
\newblock {\em Neural Information Processing Systems\/}.

\bibitem[\protect\citeauthoryear{Jiang}{Jiang}{2006}]{Jiang2006}
Jiang, W. (2006).
\newblock Bayesian variable selection for high dimensional generalized linear
  models.
\newblock {\em Ann. Statist.\/}~{\em 35}, 1487--1511.

\bibitem[\protect\citeauthoryear{Park and Hastie}{Park and
  Hastie}{2007}]{Park2007}
Park, M.~Y. and T.~Hastie (2007).
\newblock L-1 regularization path algorithm for generalized linear models.
\newblock {\em J. R. Stat. Soc. Ser. B Stat. Methodol.\/}~{\em 69}, 659--677.

\bibitem[\protect\citeauthoryear{Tibshirani}{Tibshirani}{1996}]{Tibshirani1996}
Tibshirani, R. (1996).
\newblock Regression shrinkage and selection via the lasso.
\newblock {\em J. R. Stat. Soc. Ser. B Stat. Methodol.\/}~{\em 73}, 273--282.

\bibitem[\protect\citeauthoryear{Tucker}{Tucker}{1966}]{Tucker1966}
Tucker, L. (1966).
\newblock Some mathematical notes on three-mode factor analysis.
\newblock {\em Psychometrika\/}~{\em 31}, 279--311.

\bibitem[\protect\citeauthoryear{Wu, Chen, and Hastie}{Wu
  et~al.}{2009}]{Tong2009}
Wu, T.~T., Y.~F. Chen, and T.~Hastie (2009).
\newblock Genome-wide association analysis by lasso penalized logistic
  regression.
\newblock {\em Bioinformatics\/}~{\em 25}, 714--721.

\bibitem[\protect\citeauthoryear{Yang, Wan, and Yang}{Yang
  et~al.}{2010}]{Yang2010}
Yang, C., X.~Wan, and Q.~Yang (2010).
\newblock Identifying main effects and epistatic interactions from large-scale
  {SNP} data via adaptive group lasso.
\newblock {\em BMC Bioinformatics\/}~{\em 11}.

\bibitem[\protect\citeauthoryear{Zhang and Golub}{Zhang and
  Golub}{2001}]{Zhang2001}
Zhang, T. and G.~H. Golub (2001).
\newblock Rank-one approximation to high order tensors.
\newblock {\em SIAM J. Matrix Anal. Appl.\/}~{\em 23}, 534.

\bibitem[\protect\citeauthoryear{Zou and Hastie}{Zou and
  Hastie}{2005}]{Zou2005}
Zou, H. and T.~Hastie (2005).
\newblock Regularization and variable selection via the elastic net.
\newblock {\em J. R. Stat. Soc. Ser. B Stat. Methodol.\/}~{\em 67}, 301--320.

\end{thebibliography}

\end{document}